\documentclass[12pt]{article}%
\usepackage{graphicx}
\usepackage{amsmath}
\usepackage{amsfonts}
\usepackage{amssymb}%
\providecommand{\U}[1]{\protect\rule{.1in}{.1in}}
\newtheorem{theorem}{Theorem}

\newtheorem{conjecture}[theorem]{Conjecture}

\newenvironment{proof}[1][Proof]{\textbf{#1.} }{\ \rule{0.5em}{0.5em}}
\begin{document}

\title{\textbf{Rotating states }\\\textbf{ in driven clock- and XY-models}}
\author{Christian Maes\\Instituut voor Theoretische Fysica \\K.U.Leuven, Belgium \\christian.maes@fys.kuleuven.be
\and Senya Shlosman\\Centre de Physique Th\'{e}orique UMR 6207\\Lunimy, Marseille\\shlosman@cpt.univ-mrs.fr\\Inst. of the Information Transmission Problems\\Moscow, Russia\\shlos@iitp.ru}
\maketitle

\begin{abstract}
We consider 3D active plane rotators, where the interaction between the spins
is of XY-type and where each spin is driven to rotate. For the clock-model,
when the spins take $N\gg1$ possible values, we conjecture that there are two
low-temperature regimes. At very low temperatures and for small enough drift
the phase diagram is a small perturbation of the equilibrium case. At larger
temperatures the massless modes appear and the spins start to rotate
synchronously for arbitrary small drift. For the driven XY-model we prove that
there is essentially a unique translation-invariant and stationary
distribution despite the fact that the dynamics is not ergodic.


\end{abstract}

\section{Introduction}

Understanding nonequilibrium phase transitions is a major challenge of
statistical physics, bearing many different aspects. Today not much of a
systematic theory exists, with few experimental and even fewer mathematical
results. One important question of this multifaceted subject is to describe
the changes to the equilibrium phase diagram when a steady nonequilibrium
driving is added. The best known examples are driven diffusive lattice gases
-- like boundary driven and asymmetric exclusion processes, see e.g. the books
\cite{SZ,DM}, including many results of computer simulations. The
nonequilibrium there originates from installing differences in chemical
potentials or from adding nonconservative external fields, and follows the
prescription of local detailed balance, \cite{KLS}. In the present paper a
uniform nonequilibrium force drives the \textit{internal} degrees of freedom,
the planar rotating spins. The spins take values on the unit circle and are
placed on the sites of a regular lattice. When mutually uncoupled, all spins
undergo the same non-reversible Markov evolution with a bias in the direction
of rotation. The interaction couples nearest neighbors, $x\sim y$, with energy
following the XY-model,
\begin{equation}
H(\varphi)=-\sum_{x\sim y}\cos\left(  \varphi_{x}-\varphi_{y}\right)
\label{01v}%
\end{equation}
for \textquotedblleft angles\textquotedblright\ $\varphi_{x},x\in
\mathbb{Z}^{D}$. The possible values of the angles\ determine the nature of
the spins. A first choice is to take $\varphi_{x}=2\pi k/N,k=1,2,\ldots,N$ on
the discrete circle with $N$ possible values. For $N=2$ that is the Ising
model; for $N=3$ it is equivalent with the $q=3$ Potts model. The second
possible choice is formally obtained in the $N\uparrow+\infty$ limit, and has
a continuum of values $\varphi_{x}\in\lbrack0,2\pi]$. That truly corresponds
to the XY-model where the spins are plane rotators having unit length.
\newline
The purpose of the paper is to discuss the modification of the phase diagram
when a nonequilibrium driving is inserted that induces biased rotation of the
spins over the circle. By doing so the following phenomena can be
addressed:\newline a) the uniqueness of the stationary distribution
accompanied by\textit{ breakdown of ergodicity} -- in the sense that some
initial data do not relax;\newline b) the presence of macroscopic dynamical
coherence;\newline c) the stability of equilibrium phases against small
nonequilibrium driving;\newline We first briefly introduce each of these
points, to realize them more concretely in later sections.

\subsection{Unique stationary distribution without ergodicity}

Are there stochastic dynamics with a unique stationary distribution, which are
not ergodic? Without any further restrictions this question is easy and not
very interesting, with the answer being `yes'. For discrete time a simple
example is given by the two-state Markov chain
\begin{equation}
+1\rightarrow-1\rightarrow+1\rightarrow\ldots,\label{pm}%
\end{equation}
flipping deterministically at each time. Similarly, in continuous time we can
take the rotation over the circle with a constant angular speed $v$:
\begin{equation}
\theta(t)=\theta(0)+vt\ \operatorname{mod}2\pi\label{colrot}%
\end{equation}
One wonders whether one can construct non-degenerate random processes, which
exhibit the above prototypical behavior. Of course, we must then consider
infinite-volume interacting particle systems and infinite probabilistic
cellular automata, since finite state non-degenerate Markov processes are
always ergodic.

For the case of probabilistic cellular automata (discrete time parallel
updating of spins), a construction, mimicking the behavior \eqref{pm} is
presented in a recent paper by \cite{CM}. However, the example of \cite{CM}
still has some degeneracy, because for every time $T$ one can present two
local events, $A$ and $B,$ such that the transition probability $p_{T}\left(
A|B\right)  $ in $T$ steps vanishes. Thus, we feel that a truly non-degenerate
discrete time example is still missing. We believe that the discrete time
version of our 3D driven clock model gives such non-degenerate example.

Our constructions below present the case of rotating interacting spins in
continuous time. The rotation speed $v$ in \eqref{colrot} will be induced by
the nonequilibrium driving, and the angle $\theta$ should be thought of as the
order parameter, or collective phase, of the model. The dynamics will be
nondegenerate, as local fluctuations in the phase are allowed. See Theorem
\ref{maint} for the precise result.

\subsection{Macroscopic coherence}

Not surprisingly, the mechanism above connects with the old but still not
completely resolved question of whether macroscopic dynamical coherence or
pattern formation in spatially extended systems can be obtained by local
translation-invariant and non-degenerate updating of spins and whether that is
even possible for a continuous time (sequential) dynamics and for dynamics
that satisfy detailed balance, \cite{DK,DM}. In fact, an example with that
flavor was recently described in \cite{RSV}. There an infinite queuing network
was considered, with several types of clients and with exponential service
times. In the high load regime the system exhibits coherent behavior. That
means that if the initial state of the network is close to the `coherent' one,
characterized by a given value of the `phase' observable (which takes values
on the circle), then in the process of evolution this phase evolves with a
constant speed, is never `forgotten', and the initial synchronization is never
broken. Still the system has a unique stationary distribution, with the phase
being uniformly distributed over the circle. In the language of queuing
networks it is an example of violation of the Poisson hypothesis. Yet, this
example lives not on a lattice with short range interactions, but on a
mean-field graph (which is, in some sense, an infinite complete
graph).\newline There is also a vast literature on the emergence of
synchronized rotators using variants of the so called Kuramoto model; for a
review, see \cite{AB}. Recently, a mean field analysis for active rotator
models was carried out in \cite{GG}, which for some choice of the drift is the
mean field version of the model we consider later in \eqref{coudriv}. We
believe that rotating states emerge in low temperature uniformly driven
$N$-clock models, if $N$ is sufficiently large. See the conjectures in Section
\ref{clock}.

\subsection{Stability of equilibrium phases}

When the equilibrium model has finitely many macroscopic phases in some regime
of its parameters, then we expect that these remain in place for small
driving. Below a critical driving, the changes in basin of attraction and in
macroscopic appearance will be small. To understand the nature of that
critical driving, one must realize that the stability of equilibrium phases
requires some sort of free energy barriers, or, in particle language, the
excitations must be massive. Therefore, the presence of soft modes, or
Goldstone bosons, can break stability. In the context above, that means that
the critical driving (the minimum we need to truly disturb the phase diagram)
will go down as $N\uparrow+\infty$. See conjectures \ref{h2}--\ref{h3} for
more precise speculations.\newline

The following Section \ref{clockdef} contains the details of the
nonequilibrium model -- driven $N$-clock models -- together with a summary of
the situation in equilibrium. Section \ref{clock} is devoted to what we
believe happens for finite $N$; these are mostly a collection of conjectures
in which we firmly believe but where the proofs are missing. Some of this is
remedied in Section \ref{cir} for the driven XY-model where the picture is
more complete. The main result is that the 3D driven XY-model shows
nonergodicity, while having a unique stationary translation invariant
distribution, at (almost) all low temperatures.

\section{The $N$-Clock Model}

\label{clockdef}

The $N$-clock model is an interacting particle system that lives on
$\mathbb{Z}^{3}.$ At each site $x\in\mathbb{Z}^{3}$ there is a spin
$\sigma_{x}\in\mathbb{Z}_{N},$ where $\mathbb{Z}_{N}\subset\mathbb{S}%
^{1}\subset\mathbb{C}^{1}$ is the group of $N$-th roots of unity. Each spin
$\sigma_{x}$ has its clock, and when the clock rings, the spin jumps to one of
the two `nearest' values: $\sigma_{x}\rightarrow\sigma_{x}^{\pm}=\exp\left\{
\pm\frac{2\pi i}{N}\right\}  \sigma_{x}.$ That is equivalent to introducing
the angles $\varphi_{x}$ with $\sigma_{x}=\exp i\varphi_{x},$ $\ \varphi
_{x}=2\pi k/N,$ $k=1,2,\ldots,N$, and moves $\varphi_{x}\rightarrow\varphi
_{x}\pm2\pi/N$. In what follows we will use the notation $\zeta\left(
\varphi\right)  $ for $\exp\left\{  i\varphi\right\}  ,$ $\zeta\left(
\varphi\right)  \in\mathbb{S}^{1}\subset\mathbb{C}^{1},$ so in particular
\[
\sigma_{x}\rightarrow\sigma_{x}^{\pm}=\zeta\left(  \pm\frac{2\pi}{N}\right)
\sigma_{x}.
\]
The particles are interacting, with the energy given by \eqref{01v}. We define
the rates $c\left(  x,\sigma,\pm\right)  $ of the jumps $\sigma_{x}%
\rightarrow\sigma_{x}^{\pm}$ of the spin $\sigma_{x}$ in the environment
$\sigma$ at inverse temperature $\beta$ by
\begin{equation}
c\left(  x,\sigma,\pm\right)  =p_{\pm}\exp\left\{  \frac{\beta}{2}\sum_{y\sim
x}\left[  \cos\left(  \varphi_{x}-\varphi_{y}\pm\frac{2\pi}{N}\right)
-\cos\left(  \varphi_{x}-\varphi_{y}\right)  \right]  \right\}  \label{02}%
\end{equation}
(sum over nearest neighbors $y$ of $x\in\mathbb{Z}^{3}$), where the numbers
$p_{+}\geq p_{-}>0$ are two extra parameters. Their difference is measured by
$d:=\log\big(p_{+}/p_{-}\big)\geq0$ and is called the drift. One can imagine
it as the coupling of the planar rotator with a magnetic field that acts
perpendicular to the plane.
We call the above model the `$N$-Clock model' with a drift. (Of course, the
drift does not make sense for $N=2$ (Ising model).)

We first describe the properties of the symmetric Clock model, when the drift
$d=0.$ Note that in this case the evolution defined above satisfies detailed
balance. That is what we call the equilibrium or symmetric Clock model. Then,
the Gibbs measures for \eqref{01v} are reversible stationary measures. We can
thus use the results of \cite{FILS}, theorem 4.6:

\begin{theorem}
\label{t1} (symmetric Clock model) There exists a value $\beta_{0}$ of the
inverse temperature, such that for every $\beta>\beta_{0}$ and every $N\geq2,$
the symmetric $N$-Clock model has at least $N$ different extremal stationary
distributions, $\left\langle \cdot\right\rangle _{\zeta_{k},\beta},$
$\zeta_{k}=\zeta\left(  \frac{2\pi k}{N}\right)  ,k=1,...,N.$ These states are
translation-invariant, exhibit long-range order and are magnetized:%
\[
\left\langle \sigma_{x}\right\rangle _{\zeta_{k},\beta}=m_{N}\left(
\beta\right)  e^{\frac{2\pi ik}{N}},\text{ with }m_{N}\left(  \beta\right)
>0\text{ for }\beta>\beta_{0}.
\]
(Here we interpret the spins $\sigma_{x}$ as elements of $\mathbb{C}^{1}.$)
\end{theorem}

Observe that the finite $\beta_{0}$ above remains the same for all $N\geq2$,
which makes the Theorem unreachable for the standard low-temperature analysis
based on the Peierls condition. For example, the Pirogov-Sinai theory
\cite{PS, S} would establish the stability of the $N$ ground states only for
$\beta\geq\beta_{N}^{PS},$ where $\beta_{N}^{PS}\rightarrow\infty$ as
$N\rightarrow\infty.$ The reason for that is not purely technical: indeed, in
the domain of validity of the Pirogov-Sinai theory one necessarily has
additional properties of the pure phases, such as the exponential decay of the
truncated correlation functions. However, we believe that in reality such
exponential decay holds only for low enough temperatures, $\beta>\beta_{N}%
^{G},$ with $\beta_{N}^{G}\rightarrow\infty$ as $N\rightarrow\infty,$ so the
PS-method can not be improved to reach $\beta_{0}.$ Moreover, we think that
the following is true:

\begin{conjecture}
\label{h1} There exists a value $\bar{N},$ such that for each $N\geq\bar{N}$
the 3D symmetric $N$-clock model undergoes two phase transitions. Namely, for
all $\beta>\beta_{N}^{G}$ it has $N$ pure magnetized phases, with exponential
decay of truncated correlations, with $\beta_{N}^{G}\rightarrow\infty$ as
$N\rightarrow\infty.$ For smaller intermediate values of $\beta,$ $\beta
_{N}^{G}>\beta>\beta_{N}^{cr}$ it also has at least $N$ pure phases with
non-zero magnetization; however, the correlation decay in these phases is only
algebraic. (A stronger recent conjecture of \cite{EKO} even talks about a
\textit{continuum} of pure phases, $\left\langle \cdot\right\rangle
_{\zeta,\beta},$ $\zeta\in\mathbb{S}^{1}.$) Finally, for $\beta<\beta_{N}%
^{cr}$ the model has only one Gibbs state, again with exponential decay of correlations.
\end{conjecture}

We know from \cite{ABF} that there is no intermediate phase for the Ising
model, hence $\bar{N}>2$. The conjectured behavior is somewhat similar to the
one for 2D Clock models; it was proven in \cite{FS} that these indeed undergo
a Berezinskii-Kosterlitz-Thouless phase transition. There, massless Goldstone
modes appear (hence our notation $\beta_{N}^{G}$).

\section{Conjectures for the driven Clock-model}

\label{clock}

We now discuss the situation with non-zero drift $d,$ (where we do not have
detailed balance).
Let us run our $N$-Clock model with a drift for a time duration $T,$ starting
in one of the equilibrium phases $\langle\cdot\rangle_{\zeta_{k},\beta}.$ Let
us denote the resulting state by $\left\langle \cdot\right\rangle _{\beta
,d}^{k,T}.$

\begin{conjecture}
\label{h2} For every $N\geq\bar{N},$ $\beta>\beta_{N}^{cr}$ there exists a
critical value $d_{cr}\left(  \beta,N\right)  $ of the drift, $0<d_{cr}\left(
\beta,N\right)  ,$ such that the following holds:

\begin{enumerate}
\item if $|d|<d_{cr}\left(  \beta,N\right)  ,$ then the state $\left\langle
\cdot\right\rangle _{\beta,d}^{k,T}$ approaches the state $\left\langle
\cdot\right\rangle _{\beta,d}^{k},$ as $T\rightarrow\infty$, which is
magnetized:
\[
\left\langle \sigma_{0}\right\rangle _{\beta,d}^{k}\neq0
\]
and is close to $\langle\cdot\rangle_{\zeta_{k},\beta}$ for small $d,$ in the
sense of expectations of local observables;

\item if $d>d_{cr}\left(  \beta,N\right)  ,$ then the state $\left\langle
\cdot\right\rangle _{\beta,d}^{k,T}$ is a `rotating' state as $T\rightarrow
\infty$ (in particular, it has no limit as $T\rightarrow\infty$). Namely,
there exist two periodic functions: $m\left(  T\right)  =m\left(
T;\beta,N,d\right)  >0\ $ and $\Phi\left(  T\right)  =\Phi\left(
T;\beta,N,d\right)  ,$ i.e.%
\[
m\left(  T+\omega\right)  =m\left(  T\right)  ,\ \Phi\left(  T+\omega\right)
=\Phi\left(  T\right)  ,
\]
with period $\omega$ being the mean angular velocity, $\omega=\omega\left(
\beta,N,d\right)  ,$ and a phase shift $\phi_{k}=\phi_{k}\left(
\beta,N,d\right)  ,$ such that
\begin{equation}
\left\vert \left\langle \sigma_{x}\right\rangle _{\beta,d}^{k,T}-m\left(
T\right)  e^{i\left(  \Phi\left(  T\right)  +\phi_{k}\right)  }\right\vert
\rightarrow0\text{ as }T\rightarrow\infty. \label{ro}%
\end{equation}
(Here we again are treating the spin $\sigma_{x}$ as belonging to
$\mathbb{C}^{1}.$)
\end{enumerate}
\end{conjecture}

It is interesting to compare the curve of the states $\left\langle
\cdot\right\rangle _{\beta,d}^{k,T},$ $T\geq0,$ with the conjectured
(\cite{EKO}) pure phases of the symmetric clock-model, $\left\langle
\cdot\right\rangle _{\zeta,\beta},$ $\zeta\in\mathbb{S}^{1}.$ One cannot be
stopped from guessing that perhaps for every $T$ we have $\left\langle
\cdot\right\rangle _{\beta,d}^{k,T}=\left\langle \cdot\right\rangle
_{\zeta\left(  T\right)  ,\beta},$ for some $\zeta\left(  T\right)
\equiv\zeta\left(  T,k,d,\beta\right)  \in\mathbb{S}^{1}.$

The next conjecture deals with the behavior of the critical drift $d_{cr}$
introduced above.

\begin{conjecture}
\label{h3} The critical drift is positive at low temperatures. It decreases to
become zero at $\beta_{N}^{G}$:
\[
d_{cr}(\beta,N)=0\text{ for }\beta_{N}^{cr}<\beta<\beta_{N}^{G}.
\]

\end{conjecture}

The rationale behind this conjecture is that at temperatures above $\left(
\beta_{N}^{G}\right)  ^{-1}$ the $N$-Clock model enters into the spin-wave
phase or Goldstone modes regime and so qualitatively should behave like the
$XY$-model, which is in the rotating phase for any non-zero value of drift,
see below. This similarity of the intermediate phases with the $XY$-model is
the basis of all our speculations. Hence, we believe in the following

\begin{conjecture}
\label{h4} For every $\beta$ large enough there exists $N=N\left(
\beta\right)  ,$ such that for any $N\geq N\left(  \beta\right)  $ and for all
$d>0$ the $N$-Clock model with a drift $d$ has a continuum of different
rotating states.
\end{conjecture}

One might wonder whether there is a difference between the structure of the
stationary states in the Pirogov-Sinai regime and in the Goldstone modes
regime, when $d>d_{cr}\left(  \beta,N\right)  ,$ i.e., when we are in the
regime of rotating states. We expect the answer to be positive:

\begin{conjecture}
\label{h5}

\begin{enumerate}
\item \textbf{Rotating soft modes.} In the regime $\beta\in\left(  \beta
_{N}^{cr},\beta_{N}^{G}\right)  $, $d>0$ there is a unique stationary
distribution, $\left\langle \cdot\right\rangle _{\beta,d}^{st}$. It is
translation invariant and has zero magnetization. It is given by the limit%
\begin{equation}
\left\langle \cdot\right\rangle _{\beta,d}^{st}=\lim_{T\rightarrow\infty}%
\int_{T}^{T+\omega}\left\langle \cdot\right\rangle _{\beta,d}^{k,T}%
\ dT\label{13}%
\end{equation}
(which does not depend on $k$).

\item \textbf{Rotating PS. }In the regime $\beta>\beta_{N}^{G},$
$d>d_{cr}\left(  \beta,N\right)  $, in addition to the time-stationary
translation invariant state $\left(  \ref{13}\right)  $ there are also
time-stationary non-translation invariant states (the `Dobrushin' states).
They are given by the same formula $\left(  \ref{13}\right)  ,$ where instead
of the states $\left\langle \cdot\right\rangle _{\beta,d}^{k,T}$ one should
use the states $\left\langle \cdot\right\rangle _{\beta,d}^{\pm k,T}.$ The
latter are obtained by starting the driven $N$-Clock dynamics with the measure
$\delta_{\pm k}$ that gives weight $1$ to the configuration%
\[
\sigma_{x=\left\{  x_{1},x_{2},x_{3}\right\}  }=\left\{
\begin{array}
[c]{cc}%
e^{2\pi ki/N} & \text{ for }x_{1}\geq0\\
e^{2\pi\left(  k+\frac{N}{2}\right)  i/N} & \text{ for }x_{1}<0
\end{array}
\right.  ,
\]
and where we suppose for simplicity that $N$ is even.
\end{enumerate}
\end{conjecture}

The Dobrushin time-stationary non-translation invariant states have a rigid
interface at the level $x_{1}=0$. In a typical configuration drawn from such a
state the spins on different sides of the interface are pointing in
(approximately) opposite directions, though the direction itself can be
arbitrary. One should remember here that no Dobrushin states exist in the 3D
XY-model, as shown in \cite{FP}; see also the discussion in \cite{SV}. In
fact, it is argued in \cite{FP} that there are no non-translation invariant
states at all in the 3D XY-model. This is the basis of our Conjecture
\ref{h5}. More precisely, we believe that in all cases there is a unique
translation invariant stationary distribution for $\beta>\beta_{N}%
^{cr},d>d_{cr}(\beta,N)$.

\section{The three dimensional $XY$-model: main result}

\label{cir} The dynamical $XY$-model with a drift -- called $dXY$ model below
-- can be obtained from the energy $\left(  \ref{01v}\right)  ,$ and the
dynamics $\left(  \ref{02}\right)  $ by taking the limit $N\rightarrow\infty,$
in a diffusive rescaling of time by $t\rightarrow t/N^{2}$. Alternatively, it
is a 3D model of coupled Brownian motions $\varphi_{x},x\in\mathbb{Z}^{3}$, on
circles $\varphi_{x}\in\mathbb{S}^{1}\subset\mathbb{C}.$ The Brownian motions
$\varphi_{x}$ have a constant drift, $d,$ and they are interacting via the
nearest neighbor attraction $\left(  \ref{01v}\right)  .$ The dynamics for
$\varphi_{x}(t)\in\lbrack0,2\pi]$ is then as follows: modulo $2\pi$,
\begin{equation}
\text{d}\varphi_{x}(t)=d\ \text{d}t-\frac{\partial H}{\partial\varphi_{x}%
}\ \text{d}t+\sqrt{\frac{2}{\beta}}\,\text{d}W_{x}(t), \label{coudriv}%
\end{equation}
where the $W_{x}(t)$ are independent standard Wiener processes and
\[
\frac{\partial H}{\partial\varphi_{x}}=\sum_{y:y\sim x}\sin(\varphi
_{x}(t)-\varphi_{y}(t))
\]
The formal generator of this process, acting on local smooth functions $f$ is
\begin{align}
Lf  &  =L_{0}f+d\,\sum_{x}\frac{\partial f}{\partial\varphi_{x}}%
\nonumber\label{gene}\\
L_{0}f  &  =\sum_{x}\big[-\frac{\partial H}{\partial\varphi_{x}}%
\,\frac{\partial f}{\partial\varphi_{x}}+\frac{1}{\beta}\,\frac{\partial^{2}%
f}{\partial\varphi_{x}^{2}}\big]
\end{align}
Observe that $L_{0}$ commutes with the new generator for the driven model:
$[L,L_{0}]=0$ because for all $x$,
\[
\sum_{y}\frac{\partial^{2}H}{\partial\varphi_{x}\partial\varphi_{y}}=0
\]
Of course $L-L_{0}$ generates independent rotation on each angle with angular
speed $d$ and thus commutes with the generator $L_{0}$ of the (undriven)
$XY-$model. In other words, the dynamics of the $XY-$model can be interchanged
with uniform rotation of all spins.\textbf{ }

Much of the equilibrium structure of the $XY$-model is known. At low
temperatures $\beta^{-1}$ the 3D $XY$-model has a continuum of
translation-invariant Gibbs states. They can be obtained as thermodynamic
limits $\left\langle \cdot\right\rangle _{\zeta,\beta}$ of the finite-volume
Gibbs states with coherent boundary conditions $\varphi_{x}\equiv\zeta\in$
$\mathbb{S}^{1}.$ outside the volume. These translation invariant states have
non-zero spontaneous magnetization,
\[
\left\langle \varphi_{0}\right\rangle _{\zeta,\beta}=m\left(  \beta\right)
\zeta,\ \text{with }m\left(  \beta\right)  >0,
\]
see \cite{FSS}.

Here is our main result. Consider the set $\mathcal{S}$ of stationary and
translation-invariant distributions for the $dXY-$model.

\begin{theorem}
\label{maint} The set $\mathcal{S}$ is a singleton for almost all
temperatures, while at sufficiently low temperatures there exist rotating
states as in Conjecture \ref{h2}.2: $\langle\cdot\rangle_{\beta,d}^{\zeta
,T}=\langle\cdot\rangle_{\zeta+dT,\beta}$.
\end{theorem}

\begin{proof}
We start by repeating that the evolution of the random variables $\psi
_{x}\left(  t\right)  =\varphi_{x}\left(  t\right)  -d\,t$ is that of the
$XY$-model with zero drift. If therefore $\mu$ is a stationary distribution
for the $dXY-$model, then $\mu$ is periodically repeated under the symmetric
$XY$-dynamics. Suppose now that $\mu$ is translation-invariant. Then, Holley's
argument shows that $\mu$ is in fact a translation invariant Gibbs measure for
the XY-model, \cite{Ho}. Moreover, $\mu$ must then be rotation invariant
($\mathbb{S}^{1}$-invariant) since stationary states of the $dXY-$model have
zero magnetization. From \cite{Pf} it then follows that $\mu$ is unique for
almost all temperatures.\newline On the other hand, at low temperatures
$\beta^{-1}$ the 3D $XY$-model has a continuum of translation-invariant Gibbs
states. These phases of the $XY$-model correspond in an evident way to
rotating states of the $dXY-$model. These rotating states then do not converge
to a stationary state of the $dXY$ model.
\end{proof}

\noindent\textbf{Remarks}:\newline1. The Gibbs field $\left\langle
\cdot\right\rangle _{\beta}^{st}$ , defined by
\[
\left\langle \cdot\right\rangle _{\beta}^{st}=\int_{\mathbb{S}^{1}%
}\left\langle \cdot\right\rangle _{\zeta,\beta}d\zeta
\]
is $\mathbb{S}^{1}$-invariant. We believe that for all $\beta$ the state
$\left\langle \cdot\right\rangle _{\beta}^{st}$ is the only
translation-invariant Gibbs state of the $XY$-model, which is $\mathbb{S}^{1}%
$-invariant, which would remove the \textquotedblleft almost
all\textquotedblright. (This statement is proven to hold for the $XY$-model
for almost all values of $\beta,$ \cite{FP}.)\newline

\noindent2. Note, however, that by changing the interaction from $\cos\left(
\varphi_{x}-\varphi_{y}\right)  $ to $\left(  \frac{1+\cos\left(  \varphi
_{x}-\varphi_{y}\right)  }{2}\right)  ^{p}$ -- i.e. by passing to the
so-called `very-nonlinear $\sigma-$model' -- we obtain an example of a system
which at some temperatures has at least two translation-invariant
$\mathbb{S}^{1}$-invariant Gibbs states (for $p$ large enough), see \cite{ES}.

\bigskip

\noindent\textbf{Acknowledgment.} We are grateful to Aernout van Enter for
useful comments. S.S. acknowledges the financial support of Instituut voor
Theoretische Fysica, K.U.Leuven, Belgium during his visit in May, 2011, where
part of this work was done.


\begin{thebibliography}{99}                                                                                               %

\bibitem {AB}J.A.~Acebr\'on, L.L.~Bonilla, C.J.~P\'erez Vicente, F.~Ritort and
R.~Spigler, The Kuramoto model: A simple paradigm for synchronization
phenomena, Rev. Mod. Phys. \textbf{77}, 137--185 (2005).

\bibitem {ABF}M.~Aizenman, D.J.~Barsky and R.~Fern\'{a}ndez, The phase
transition in a general class of Ising-type models is sharp, J. Stat. Phys.
\textbf{47}, 343--374 (1987).

\bibitem {CM}P.~Chassaing and J.~Mairesse, A non-ergodic probabilistic
cellular automaton with a unique invariant measure, \texttt{arXiv:1009.0143v2
[cs.FL]}.

\bibitem {DK}M.~Diakonova and R.S.~MacKay, Mathematical examples of space-time
phases, to appear in Int. J. Bif. Chaos.

\bibitem {DM}R.~Dickman and J.~Marro, \textit{Nonequilibrium Phase Transitions
in Lattice Models}, Cambridge University Press (May 13, 1999).

\bibitem {EKO}A.C.D.~van Enter, C.~K\"ulske and A.A.~Opoku, Discrete
approximations to vector spin models, \texttt{arXiv:1104.4241v1 [math-ph]}.

\bibitem {ES}A.C.D.~van Enter, S.B.~Shlosman, Provable first-order transitions
for liquid crystal and lattice gauge models with continuous symmetries, Comm.
Math. Phys. \textbf{255}, 21--32 (2005).

\bibitem {FSS}J.~Fr\"{o}hlich, B.~Simon and T.~Spencer, Infrared bounds, phase
transitions and continuous symmetry breaking, Comm. Math. Phys. \textbf{50},
79--95 (1976).

\bibitem {FILS}J.~Fr\"{o}hlich, R.~Israel, E.H.~Lieb and B.~Simon, Phase
transitions and reflection positivity. I. General theory and long range
lattice models, Comm. Math. Phys. \textbf{62}, 1--34 (1978).

\bibitem {FP}J.~Fr\"ohlich and C.-E.~Pfister, Spin waves, vortices, and the
structure of equilibrium states in the classical XY model, Comm. Math. Phys.
\textbf{89}, 303--327 (1983).

\bibitem {FS}J.~Fr\"{o}hlich and T.~Spencer, The Kosterlitz-Thouless
transition in two-dimensional Abelian spin systems and the Coulomb gas, Comm.
Math. Phys. \textbf{81}, 527--602 (1981).

\bibitem {GG}G.~Giacomin, K.~Pakdaman, X.~Pellegrin and C.~Poquet, Transitions
in active rotator systems: invariant hyperbolic manifold approach,
\texttt{arXiv:1106.0758v1 [math-ph]}.

\bibitem {Gr}G.~Grinstein, D.~Mukamel, R.~Seidin, and C.H.~Bennett, Temporally
periodic phases and kinetic roughening, Phys. Rev. Lett. \textbf{70},
3607---3610 (1993).

\bibitem {Ho}R.~Holley, Free energy in a Markovian model of a lattice spin
system, Commun. Math. Phys. \textbf{23}, 87--99 (1971).

\bibitem {KLS}S.~Katz, J.L.~Lebowitz and H.~Spohn, Phase Transitions in
Stationary Non-equilibrium States of Model lattice Systems, Phys. Rev. B
\textbf{28}, 1655--1658 (1983).

\bibitem {Pf}C.-E.~Pfister, Translation invariant equilibrium states of
ferromagnetic abelian lattice systems, Comm. Math. Phys. \textbf{86}, 375--390 (1982).

\bibitem {PS}S.A.~Pirogov and Ya.G.~Sinai, Phase Diagrams of Classical Lattice
Systems, Theor. and Math. Phys. \textbf{25}, 358--369, 1185--1192 (1975).

\bibitem {RSV}A.~Rybko, S.~Shlosman and A.~Vladimirov, Spontaneous Resonances
and the Coherent States of the Queuing Networks, J. Stat. Phys. \textbf{134},
67--104 (2009).

\bibitem {SZ}B.~Schmittman and R.K.P.~Zia, \textit{Statistical Mechanics of
Driven Diffusive Systems}, Volume 17 (Phase Transitions and Critical
Phenomena), Cyril Domb (Series Editor), R.K.P.~Zia (Series Editor),
B.~Schmittmann (Series Editor), J.L.~Lebowitz (Series Editor), Academic Press
(August 28, 1995), pp 3--214.

\bibitem {SV}S.~Shlosman and Y.~Vignaud, Dobrushin Interfaces via Reflection
Positivity, Comm. Math. Phys. \textbf{276}, 827--861 (2007).

\bibitem {S}Ya.G.~Sinai, \textit{Theory of Phase Transitions}, Budapest:
Academia Kiado and London: Pergamon Press (1982).
\end{thebibliography}
\end{document}